\documentclass[conference]{IEEEtran}
\usepackage[margin=0.7in]{geometry}
\usepackage[english]{babel}
\usepackage{amsmath,amsthm}
\usepackage{amsfonts}
\usepackage{extarrows}
\usepackage{amssymb}
\usepackage{dblfloatfix}
\twocolumn
\newtheorem{theorem}{Theorem}
\newtheorem{corollary}{Corollary}
\newtheorem{lemma}{Lemma}

\theoremstyle{definition}

\theoremstyle{remark}
\newtheorem{remark}{Remark}

\begin{document}
\title{Key Capacity with Limited One-Way Communication for Product Sources}
\author{\IEEEauthorblockN{Jingbo Liu~~~~~~~~Paul Cuff~~~~~~~~Sergio Verd\'{u}}
\IEEEauthorblockA{Dept. of Electrical Eng., Princeton University, NJ 08544\\
\{jingbo,cuff,verdu\}@princeton.edu}}
\maketitle
\begin{abstract}
We show that for product sources, rate splitting is optimal for secret key agreement using limited one-way communication at two terminals. This yields an alternative proof of the tensorization property of a strong data processing inequality originally studied by Erkip and Cover and amended recently by Anantharam et al. We derive a `water-filling' solution of the communication-rate--key-rate tradeoff for two arbitrarily correlated vector Gaussian sources, for the case with an eavesdropper, and for stationary Gaussian processes.
\end{abstract}
\IEEEpeerreviewmaketitle
\section{Introduction}\label{sec1}
The fundamental limit on the amount of secret key (or common randomness) that can be generated by two terminals which observe correlated discrete memoryless sources was studied in \cite{csiszar2004secrecy},\cite{ahlswede1998common}, where single-letter solutions were derived for the class of protocols with one-way communication. However, in many practical applications, the key capacity is not known, since the optimizations over auxiliary random variables in those single-letter formulas are usually hard to solve. In \cite{watanabe} the fundamental limit was extended to sources with continuous alphabets, and it was shown that for vector Gaussian sources it suffices to consider auxiliary random variables which are jointly Gaussian with the sources. Thus the capacity region for vector Gaussian sources was posed as a matrix optimization problem. Still, an explicit formula for the key capacity was not derived except for scalar Gaussian sources.

In this paper we provide the explicit formula for key capacity of vector Gaussian sources by considering what turns out to be more general: the key capacity of arbitrary product sources. Specifically, suppose terminals A, B and an eavesdropper observe discrete memoryless vector sources $\mathbf{X}=(X_i)_{i=1}^n$, $\mathbf{Y}=(Y_i)_{i=1}^n$ and $\mathbf{Z}=(Z_i)_{i=1}^n$ respectively, where
\begin{align}\label{assump}
P_{\mathbf{XY}}=\prod_{i=1}^{n} P_{X_iY_i},\quad P_{\mathbf{XZ}}=\prod_{i=1}^{n} P_{X_iZ_i}.
\end{align}
We call a source of this kind a \emph{product source} because of the structure of its joint probability distribution. The maximal rate of secret key achievable as a function of public communication rate $r$ from A to B is denoted as $R(r)$. We show that
\begin{align}\label{disp}
R(r)=\max_{\sum_{i=1}^n r_i\le r}\sum_{i=1}^n R_{i}(r_{i}),
\end{align}
where $R_{i}(r_{i})$ is the maximal achievable key rate given a communication rate of $r_i$ for the $i$'th source triple: $(X_i,Y_i,Z_i)$. This is analogous to a result in rate distortion theory about rate for a product source with a sum distortion measure \cite{shannon1959coding}, where Shannon showed that the rate distortion function is derived by summing the rates and distortions of points in the curves with the same slope.

In the case of jointly vector Gaussian sources without an eavesdropper (or with an eavesdropper but under suitable conditions), one can always apply separate invertible linear transforms on the vectors observed at A and B so that the source distribution is of the form in (\ref{assump}), thus deriving an explicit formula of $R(r)$ utilizing corresponding results of scalar Gaussian sources. The solution displays a `water filling' behavior similar to the rate distortion function of vector Gaussian sources \cite{cover2012elements}. We shall also discuss the extension to stationary Gaussian processes.

There is a curious connection between our result about $R(r)$ for product sources and the tensorization property of a strong data processing inequality originally studied by Erkip and Cover and amended recently by Anantharam et al. \cite{anantharam2013maximal}. Suppose $P_{XY}$ is given. In \cite{erkip1998efficiency} it was mistakenly claimed that $s^*(X;Y):=\sup_{U-X-Y,I(U;X)\neq0}\frac{I(U;Y)}{I(U;X)}=\rho_m^2(X;Y)$, where $\rho_m^2(X;Y)$ denotes the maximal correlation coefficient \cite{witsenhausen1975sequences}. In fact, \cite{anantharam2013maximal} shows that this is not always true and gives a general while less explicit simplification for $s^*(X;Y)$:
\begin{align}\label{ssup}
s^*(X;Y)=\sup_{Q_X\neq P_X}\frac{D(Q_Y||P_Y)}{D(Q_X||P_X)}.
\end{align}
Nevertheless, $\rho_m^2(X;Y)$ and $s^*(X;Y)$ do agree for some simple distributions of $P_{XY}$ such as Gaussian and binary with equiprobable marginals, and they both fulfill an important property called tensorization. Moreover, they are closely linked to the problem of key generation \cite{witsenhausen1975sequences}\cite{zhao}.\footnote{For the reason we just discussed, the $\rho^2_m(X;Y)$ in the expressions of efficiency functions in \cite{zhao} should be replaced by $s^*(X;Y)$.} To add one more connection between $s^*(X;Y)$ and key generation, we demonstrate that (\ref{disp}) implies the tensorization property of $s^*(X;Y)$.


\section{Preliminaries}\label{sec2}
Consider the model explained in Section \ref{sec1} with the sources distributed as in (\ref{assump}). We use $N$ to denote block length since $n$ has been reserved for the length of the vector in (\ref{assump}). The standard notation $\mathbf{X}_1^N=\mathbf{X}^N$ will be used for the block $(\mathbf{X}_i)_{i=1}^N$.

Upon receiving $\mathbf{X}^N$, terminal A computes an integer $K_1=K_1(\mathbf{X}^N)\in\mathcal{K}_1$ and sends a message $W=W(\mathbf{X}^N)\in\mathcal{W}$ through a noiseless public channel to terminal B. Then B computes the key $K_2=K_2(W(\mathbf{X}^N),\mathbf{Y}^N)$ based on available information. A rate pair $(r,R)$ is said to be achievable if it can be approached by a sequence of schemes satisfying the following conditions on the probability of agreement and security:
\begin{align}
\lim_{N\to\infty}\mathbb{P}(K_1\neq K_2)&=0,\\
\lim_{N\to\infty}[\log|\mathcal{K}_1|-H(K_1|W,\mathbf{Z}^N)] &=0,\\
\limsup_{N\to\infty}\frac{1}{N}\log|\mathcal{W}|&\le r,\\
\liminf_{N\to\infty}\frac{1}{N}\log|\mathcal{K}_1|&\ge R.
\end{align}
The function
\begin{align}
R(r):=\sup\{R:(r,R)\textrm{ is achievable}\}
\end{align}
characterizes the maximal possible key rate with a certain public communication rate.

\section{Main Results}
\subsection{Initial Efficiency in Key Generation}
Fix $P_{XYZ}$. Define
\begin{align}
&s^*_Z(X;Y)\nonumber\\
=&\sup_{U,V}\frac{I(V;Y|U)-I(V;Z|U)}
{I(V;X|U)-I(V;Z|U)+I(U;X)-I(U;Y)},\label{ss}
\end{align}
where the supremum is over all $(U,V)$ such that $(U,V)-X-(Y,Z)$ and that the denominator in (\ref{ss}) doesn't vanish. Note that the denominator is always nonnegative; if it vanishes for all $U,V$, then so does the numerator and we set $s^*_Z(X;Y)=0$. $s^*_Z(X;Y)$ is closely related to the `initial efficiency' of key generation, defined as $\lim_{r\to 0}\frac{R(r)}{r}$. In the special case of no eavesdropper, this is related to the result in \cite{zhao}, which uses the incorrect constant $\rho^2_m(X;Y)$ as we mentioned earlier. The following result gives some basic properties of $s^*_Z(X;Y)$.
\begin{theorem}\label{thm1}~\\
1) $0\le s^*_Z(X;Y)\le 1$.\\
2) $s^*_Z(X;Y)$ tensorizes, in the sense that for product sources,
\begin{align}
s^*_{Z^n}(X^n;Y^n)=\max_{1\le i\le n}s^*_{Z_i}(X_i;Y_i).
\end{align}
3) $s^*_Z(X;Y)$ is linked to the initial efficiency of key generation by
\begin{align}
\lim_{r\to 0}\frac{R(r)}{r}=\sup_{r> 0}\frac{R(r)}{r}=\frac{s^*_Z(X;Y)}{1-s^*_Z(X;Y)}.
\end{align}
4)
\begin{align}\label{e18}
&s^*_Z(X;Y)\nonumber\\
=&\sup_{Q_{VX}}\frac{I(\bar{V};\bar{Y})-I(\bar{V};\bar{Z})}
{I(\bar{V};\bar{X})-I(\bar{V};\bar{Z})+D(Q_X||P_X)-D(Q_Y||P_Y)},
\end{align}
where $\bar{V},\bar{X},\bar{Y},\bar{Z}$ has joint distribution $P_{\bar{V}\bar{X}\bar{Y}\bar{Z}(v,x,y,z)}=Q_{VX}(v,x)P_{YZ|X}(y,z|x)$. The supremum is over all $Q_{VX}$ such that the above denominator does not vanish.

Computation can be further simplified for degraded sources $X-Y-Z$:
\begin{align}
s^*_Z(X;Y)=&\sup_{U}\frac{I(U;Y|Z)}{I(U;X|Z)}\nonumber\\
=&\sup_{Q_X}\frac{D(Q_Y||P_Y)-D(Q_Z||P_Z)}{D(Q_X||P_X)-D(Q_Z||P_Z)}\label{eq19}
\end{align}
where $Q_{XYZ}=Q_X P_{YZ|X}$.
\end{theorem}
\begin{proof}
1) From the data processing inequality the denominator in (\ref{ss}) is nonnegative, and $s^*_Z(X;Y)\le 1$. If there exists $U$ such that
$I(U;X)-I(U;Y)>0$, we can choose $V$ independent of $U,X,Y$ so that the numerator vanishes whereas the denominator is positive, which shows that $s^*_Z(X;Y)\ge0$. Otherwise if $I(U;X)-I(U;Y)=0$ for all $U$, the numerator will always be nonnegative:
\begin{align}
&I(V;Y|U)-I(V;Z|U)\nonumber\\
=&I(U,V;Y)-I(U,V;Z)-I(U;Y)+I(U;Z)\nonumber\\
=&I(U,V;X)-I(U,V;Z)-I(U;X)+I(U;Z).
\end{align}
Hence $s^*_Z(X;Y)\ge0$ always holds.

2) We only need to show that $s^*_Z(X^n;Y^n)\le \max_{1\le i\le n}s^*_{Z_i}(X_i;Y_i)$ since the other direction is obvious. For any $U,V$ such that $(U,V)-X_1^n-(Y_1^n,Z_1^n)$ and $I(U,V;X_1^n)-I(U,V;Y_1^n)>0$, $I(V;Y^n|U)-I(V;Z^n|U)>0$, let $U_1^n,V_1^n$ be as in Lemma \ref{lem1} in the appendix. Then
\begin{align}
&\frac{I(V;Y^n|U)-I(V;Z^n|U)}{I(U,V;X^n)-I(U,V;Y^n)}\label{leftmost}\\
\le& \frac{\sum_{i=1}^n [I(V_i;Y_i|U_i)-I(V_i;Z_i|U_i)]}{\sum_{i=1}^n [I(U_i,V_i;X_i)-I(U_i,V_i;Y_i)]}\nonumber\\
\le& \max_{i\in \mathcal{I}}\frac{I(V_i;Y_i|U_i)-I(V_i;Z_i|U_i)}
{I(U_i,V_i;X_i)-I(U_i,V_i;Y_i)}\nonumber\\
\le& \max_{1\le i\le n}\sup_{U_i,V_i}\frac{I(V_i;Y_i|U_i)-I(V_i;Z_i|U_i)}
{I(U_i,V_i;X_i)-I(U_i,V_i;Y_i)}
\end{align}
\begin{figure*}[b]
\hrulefill
\begin{align}
&\frac{I(V;Y|U)-I(V;Z|U)}
{I(V;X|U)-I(V;Z|U)+I(U;X)-I(U;Y)}\nonumber\\
=&\frac{\int[I(V;Y|U=u)-I(V;Z|U=u)]{\rm d}P_U(u)}{\int[I(V;X|U=u)-I(V;Z|U=u)
+D(P_{X|U=u}||P_X)-D(P_{Y|U=u}||P_Y)]{\rm d}P_U(u)}\label{equ18}\\
\le&\sup_{u}\frac{I(V;Y|U=u)-I(V;Z|U=u)}{I(V;X|U=u)-I(V;Z|U=u)
+D(P_{X|U=u}||P_X)-D(P_{Y|U=u}||P_Y)}\\
\le&\sup_{Q_{VX}}\frac{I(\bar{V};\bar{Y})-I(\bar{V};\bar{Z})}
{I(\bar{V};\bar{X})-I(\bar{V};\bar{Z})+D(Q_X||P_X)-D(Q_Y||P_Y)}.\label{equ20}
\end{align}
\end{figure*}
where $\mathcal{I}$ is the set of indices such that $I(U_i,V_i;X_i)-I(U_i,V_i;Y_i)\neq 0$, and the suprema are over all $U_i,V_i$ such that $(U_i,V_i)-X_i-(Y_i,Z_i)$ and $I(U_i,V_i;X_i)-I(U_i,V_i;Y_i)\neq 0$. Supremizing with respect to $U,V$ on (\ref{leftmost}) shows the tensorization property of $\frac{s^*_Z(X;Y)}{1-s^*_Z(X;Y)}$, which is equivalent to the tensorization property of $s^*_Z(X;Y)$.

3) The achievable region of $(r,R)$ is the union of
\begin{align}\label{region0}
&[I(U,V;X)-I(U,V;Y),\infty)\nonumber\\
&\times[0,I(V;Y|U)-I(V;Z|U)]
\end{align}
over all $U,V$ such that $(U,V)-X-(Y,Z)$ \cite{csiszar2004secrecy}. Thus $\sup_{r> 0}\frac{R(r)}{r}=\frac{s^*_Z(X;Y)}{1-s^*_Z(X;Y)}$ follows immediately from the definition of $s^*_Z(X;Y)$. The claim of
$
\lim_{r\downarrow0}\frac{R(r)}{r}=\sup_{r> 0}\frac{R(r)}{r}
$
follows from the convexity of the achievable rate region.

4) One direction of the inequality is obvious: if $U,V$ are such that $I(V;Y|U)-I(V;Z|U)\ge 0$, we have (see (\ref{equ18})-(\ref{equ20})).

Conversely, for any $Q_{VX}$, consider
\begin{align}
Q^{(1)}_{VXYZ}&=Q_{VX}P_{YZ|X},\\
Q^{(0)}_{VXYZ}&=P_V\cdot\frac{P_{XYZ}-\alpha Q^{(1)}_{XYZ}}{1-\alpha},\label{eq23}\\
P^{\alpha}_{XYZUV}(x,y,z,u,v)
=&(1-\alpha)Q^{(0)}_{VXYZ}(v,x,y,z)1_{u=0}\nonumber\\
&+\alpha Q^{(1)}_{VXYZ}(v,x,y,z)1_{u=1},
\end{align}
where $P_V$ is an arbitrary probability distribution on $\mathcal{V}$. We can assume that $P_{XYZ}(x,y,z)>0$ for all $x,y,z$ so that (\ref{eq23}) is a well-defined distribution for $\alpha>0$ small enough. Then, we can verify that $P^{\alpha}_{XYZ}=P_{XYZ}$, $(U,V)-X-(Y,Z)$, and
\begin{align}
&\lim_{\alpha\downarrow 0}\frac{I(V;Y|U)-I(V;Z|U)}
{I(V;X|U)-I(V;Z|U)+I(U;X)-I(U;Y)}\nonumber\\
=&\frac{I(\bar{V};\bar{Y})-I(\bar{V};\bar{Z})}
{I(\bar{V};\bar{X})-I(\bar{V};\bar{Z})+D(Q_X||P_X)-D(Q_Y||P_Y)},
\end{align}
where $P_{UVXYZ}:=P^{\alpha}_{UVXYZ}$. Thus (\ref{e18}) holds.

The proof of (\ref{eq19}) is omitted here due to space constraint.
\end{proof}
Thanks to Theorem \ref{thm1}, we can always eliminate one auxiliary r.v. if we only want to compute $s_Z^*(X;Y)$ instead of the rate region, which considerably reduces the dimension in the optimization problem. The interpretation of the tensorization of $s^*_Z(X;Y)$ is that, with small allowable public communication, it is always efficient to only use the best component of the product sources. Alternatively, the fact that rate splitting is optimal for product sources (as we shall prove in the sequel) implies the tensorization property of $s_Z^*(X;Y)$.

\subsection{Secret Key Generation from Product Sources}
We show that by appropriately splitting the communication rate to each `factor' in the product source, producing keys separately and combining them (called `rate splitting'), one can always achieve the optimal key rate.
\begin{theorem}\label{thm3}
In the problem of key generation from product sources satisfying (\ref{assump}), the maximum key rate satisfies
\begin{align}\label{eq18}
R(r)=\max_{\sum_{i=1}^n r_i\le r}\sum_{i=1}^n R_{i}(r_{i}),
\end{align}
which can be achieved by rate splitting. Here $R_{i}(r_i)$ is the key-rate--communication-rate function corresponding to the $i$'th source triple $(X_i,Y_i,Z_i)$. Further, if $R_{i}(\cdot)$ is differentiable and
\begin{align}
R(r)=\sum_{i=1}^n R_{i}(r^*_{i}),
\end{align}
then for each $i$, either $R'_{i}(r^*_{i})=\mu$ or $r^*_{i}=0$, where $\mu$ is some constant.
\end{theorem}
\begin{proof}
Each rate of $R_{i}(r^*_{i})$ can be approached by a scheme that operates on each of the $i$'th source triple separately. From the second equation in (\ref{assump}), the combination of these schemes form a legitimate scheme for the product source. Thus, the direction of $R(r)\ge\sum_{i=1}^n R_{i}(r_{i})$ is trivial.

By (\ref{region0}) the achievable region of $(r,R)$ is the union of
\begin{align}\label{region}
&[I(U,V;X_1^n)-I(U,V;Y_1^n),\infty)\nonumber\\
&\times[0,I(V;Y_1^n|U)-I(V;Z_1^n|U)]
\end{align}
over all $(U,V)$ such that $(U,V)-X_1^n-(Y_1^n,Z_1^n)$. The achievable region with rate splitting is the union of
\begin{align}
&\left[\sum_{i=1}^n I(U_i,V_i;X_i)-\sum_{i=1}^n I(U_i,V_i;Y_i),\infty\right)\nonumber\\
&\times\left[0,\sum_{i=1}^n I(V_i;Y_i|U_i)-\sum_{i=1}^nI(V_i;Z_i|U_i)\right]
\end{align}
over all $(U_i,V_i)$ such that $(U_i,V_i)-X_i-(Y_i,Z_i)$. The second union contains the first union, according to result of Lemma~\ref{lem1} in the appendix. Hence we also have $R(r)\le\sum_{i=1}^n R_{i}(r^*_{i})$ for some $r^*_{i}$, $i=1\dots n$. The last claim in the theorem for differentiable $R_{i}(\cdot)$ is from the KKT condition.
\end{proof}

As an application, we derive a `water-filling' solution for the communication-rate--key-rate tradeoff from two arbitrarily correlated vector Gaussian sources. For a nonnegative definite matrix $\bf\Sigma$, let ${\bf \Sigma}^{-1/2}$ be a positive definite matrix such that ${\bf \Sigma}^{-1/2}{\bf \Sigma\Sigma}^{-1/2}={\bf I}_r$, where ${\bf I}_{r}$ denotes the identity matrix of dimension $r={\rm rank}({\bf \Sigma})$. Also write ${\bf \Sigma}^{-1}=({\bf \Sigma}^{-1/2})^2$, which agrees with the definition of inverse matrix when $\bf \Sigma$ is invertible. Note that ${\bf \Sigma}^{-1/2}$ (and therefore ${\bf \Sigma}^{-1}$) may not be unique, although their choices do not affect the value of our key capacity expression. The following fact about Gaussian distributions is useful. The proof is based on singular value decomposition and is omitted here due to limitation of space.
\begin{lemma}\label{lem3}
Suppose $\mathbf{X},\mathbf{Y},\mathbf{Z}$ are jointly Gaussian vectors. There exist invertible linear transforms $\mathbf{X}\mapsto \bar{\mathbf{X}}$, $\mathbf{Y}\mapsto \bar{\mathbf{Y}}$, $\mathbf{Z}\mapsto \bar{\mathbf{Z}}$ such that all the five covariance matrices $\mathbf{\Sigma}_{\bar{\mathbf{X}}}$, $\mathbf{\Sigma}_{\bar{\mathbf{Y}}}$, $\mathbf{\Sigma}_{\bar{\mathbf{Z}}}$, $\mathbf{\Sigma}_{\bar{\mathbf{X}}\bar{\mathbf{Y}}}$, $\mathbf{\Sigma}_{\bar{\mathbf{X}}\bar{\mathbf{Z}}}$ are diagonalized if and only if $\mathbf{\Sigma}^{-1/2}_{\mathbf{X}}\mathbf{\Sigma}_{\mathbf{X}\mathbf{Z}}\mathbf{\Sigma}^{-1}_{\mathbf{Z}}\mathbf{\Sigma}_{\mathbf{Z}\mathbf{X}}\mathbf{\Sigma}^{-1/2}_{\mathbf{X}}$ commutes with $\mathbf{\Sigma}^{-1/2}_{\mathbf{X}}\mathbf{\Sigma}_{\mathbf{X}\mathbf{Y}}\mathbf{\Sigma}^{-1}_{\mathbf{Y}}\mathbf{\Sigma}_{\mathbf{Y}\mathbf{X}}\mathbf{\Sigma}^{-1/2}_{\mathbf{X}}$.
\end{lemma}
\begin{remark}\label{rem2}
Luckily, the commutativity assumption is satisfied by stationary Gaussian processes (asymptotically over large blocklengths). This is due to the commutativity of convolution.
\end{remark}

\begin{corollary}\label{cor2}
If $\mathbf{X},\mathbf{Y}$ are jointly Gaussian vectors, then there exist invertible linear transforms $\mathbf{X}\mapsto \bar{\mathbf{X}}$, $\mathbf{Y}\mapsto \bar{\mathbf{Y}}$ such that $\mathbf{\Sigma}_{\bar{\mathbf{X}}}$, $\mathbf{\Sigma}_{\bar{\mathbf{Y}}}$, $\mathbf{\Sigma}_{\bar{\mathbf{X}}\bar{\mathbf{Y}}}$ are diagonalized.
\end{corollary}

Thanks to Theorem \ref{thm3} and Corollary \ref{cor2}, the task of finding the key capacity of arbitrarily correlated Gaussian vector sources in the absence of eavesdroppers is reduced to the case of product Gaussian sources $(X^n,Y^n)$ satisfying (\ref{assump}). In the presence of an eavesdropper, it is not always possible to reduce the problem to the case of product sources, since the conditions in Theorem \ref{thm3} are not always fulfilled; but we discuss its practical relevance later. We now present the key capacity of product Gaussian sources. The solution displays a `water-filling' behaviour which is reminiscent of the rate-distortion function for Gaussian vectors~\cite{cover2012elements}.
\begin{theorem}\label{thm4}
If $(X^n,Y^n,Z^n)$ are product Gaussian sources, then the achievable communication and key rates are parameterized by $\mu>0$ as
\begin{align}\label{rp}
r=\frac{1}{2}\sum_{i:\beta_i>\mu}\log\frac{\beta_i(\mu+1)}{(\beta_i+1)\mu},
\end{align}
\begin{align}\label{rk}
R=\frac{1}{2}\sum_{i:\beta_i>\mu}\log\frac{\beta_i+1}{\mu+1},
\end{align}
where
$
\beta_i:=\frac{\rho_{X_iY_i}^2-\rho_{X_iZ_i}^2}{1-\rho_{X_iY_i}^2}.
$
\end{theorem}
\begin{remark}
The initial efficiency is $\max_{1\le i\le n}\beta_i^+$, where $\beta_i^+:=\max\{\beta_i,0\}$.
\end{remark}
\begin{proof}
Reference \cite{watanabe} derived an explicit formula for the achievable key rate in the case of degraded scalar sources:
\begin{align}
R(r)=& \frac{1}{2}\log\frac{\Sigma_{y|xz}2^{-2r}
+\Sigma_{y|z}(1-2^{-2r})}{\Sigma_{y|xz}}\label{scalar}\\
=& \frac{1}{2}
\log\left(\frac{1-\rho_{xz}^2}
{1-\rho_{xy}^2}+\frac{\rho_{xy}^2-\rho_{xz}^2}
{1-\rho_{xy}^2}2^{-2r}\right).\label{scalar1}
\end{align}
Note that the achievable rate region depends only on the marginal distributions $P_{XY}$ and $P_{XZ}$ (see (\ref{region0})), hence (\ref{scalar1}) holds as long as $\rho_{xy}\ge\rho_{xz}$, in which case $X,Y,Z$ is stochastically degraded \cite{cover2012elements} in that order. When $\rho_{xy}<\rho_{xz}$, from (\ref{region}) we see that $R(r)=0$ since $X,Z,Y$ is stochastically degraded in that order. Hence in all cases, we can further simplify
\begin{align}\label{scalar2}
R(r)= \frac{1}{2}
\log\left(1+\beta^+-\beta^+ 2^{-2r}\right)
\end{align}
where $\beta:=\frac{\rho_{xy}^2-\rho_{xz}^2}{1-\rho_{xy}^2}$ and recall the notation $\beta^+:=\max\{\beta,0\}$.

Now consider the product sources, and suppose $r^*_{i}$ are the numbers that achieve the maximum in Theorem \ref{thm3}.
According to Theorem \ref{thm3}, either $R_{i}'(r^*_{i})=\frac{\beta^+_i 2^{-2r^*_{i}}}{1+\beta^+_i-\beta^+_i 2^{-2r^*_{i}}}=\mu$ or $r^*_{i}=0$ for each $i$, where $\mu$ is some constant. For fixed $\mu$, this means
\begin{align}
r^*_{i}=\max\left\{0,\frac{1}{2}\log
\frac{(1+\mu)\beta^+_i}{\mu(1+\beta^+_i)}\right\}.
\end{align}
Equivalently, we can write
\begin{align}\label{rp}
r^*_{i}=\frac{1}{2}\log\frac{\beta^+_i(m_i+1)}{(\beta^+_i+1)m_i},
\end{align}
where $m_i:=\min\{\mu,\beta^+_i\}$. The claim then follows by substituting the value of $r^*_{i}$ into (\ref{scalar2}) and applying (\ref{eq18}).
\end{proof}

Theorem \ref{thm4} can be used to derive the key-rate--communication-rate tradeoff for Wiener class\footnote{We say a Gaussian process is in the Wiener class if its correlation function is absolutely summable \cite{gray2006toeplitz}. This is merely a technical condition which facilitates the proof.} stationary Gaussian processes, in which case the linear transforms in Lemma~\ref{lem3} can be easily found. We shall only discuss the basic idea here since a rigorous derivation is complicated and involves the asymptotic behaviour of Toeplitz matrices. Consider first passing $Y$ through a filter whose impulse response is $R_{XY}$,
\footnote{When $R_{XY}$ is bandlimited, convolution with $R_{XY}$ becomes a degenerate linear transform. In this case we can use a signal $\hat{R}_{XY}$ as an alternative, where $\hat{R}_{XY}$ has full spectrum and agrees with $R_{XY}$ in the pass-band of $R_{XY}$. The final formula of key capacity however will remain unchanged.}
the correlation function between $X$ and $Y$, resulting in a new process $\hat{Y}$. Similarly, construct $\hat{Z}$ by convolving with $R_{XZ}$. Note that
\begin{align}
R_{X\hat{Y}}&=R_{XY}*R_{YX},\\
R_{X\hat{Z}}&=R_{XZ}*R_{ZX},
\end{align}
become symmetric functions. Setting $\bar{\mathbf{X}}={\bf Q}X^n$, $\bar{\mathbf{Y}}={\bf Q}\hat{Y}^n$, $\bar{\mathbf{Z}}={\bf Q}\hat{Z}^n$ where $\bf Q$ the sine/cosine orthogonal matrix, the covariance matrices $\bf\Sigma_{\bar{X}}$, $\bf\Sigma_{\bar{Y}}$, $\bf\Sigma_{\bar{Z}}$, $\bf\Sigma_{\bar{X}\bar{Y}}$, $\bf\Sigma_{\bar{X}\bar{Z}}$ will be approximately diagonal (asymptotically for large $n$).

In summary, in the spectral representation the original Gaussian sources are converted sources satisfying the product assumption (\ref{assump}), and the correlation coefficients corresponding to frequency $\omega$ are
\begin{align}
\rho_{X'Y'}(\omega)=&\frac{|S_{XY}(\omega)|}{\sqrt{S_{X}(\omega)S_{Y}(\omega)}},\label{rho}\\
\rho_{X'Z'}(\omega)=&\frac{|S_{XZ}(\omega)|}{\sqrt{S_{X}(\omega)S_{Z}(\omega)}},\label{rho1}
\end{align}
where $S_{X},S_{Y},S_{Z},S_{XY},S_{XZ}$ denote the spectral densities and joint spectral densities. Using (\ref{rho}), (\ref{rho1}) and Theorem~\ref{thm4},
the key capacity of Gaussian processes is obtained as follows:
\begin{theorem}
For Wiener class stationary Gaussian processes, we have
\begin{align}
r=&\frac{1}{4\pi}\int_{\beta(\omega)>\mu}\log\frac{\beta(\omega)(\mu+1)}{(\beta(\omega)
+1)\mu}{\rm d}\omega,\\
R=&\frac{1}{4\pi}\int_{\beta(\omega)>\mu}
\log\frac{\beta(\omega)+1}{\mu+1}{\rm d}\omega.
\end{align}
where
$\beta(\omega)=\frac{|S_{XY}(\omega)|^2S_Z(\omega)-|S_{XZ}
(\omega)|^2S_Y(\omega)}{S_X(\omega)S_Y(\omega)S_Z(\omega)-
|S_{XY}(\omega)|^2S_Z(\omega)}$.
\end{theorem}
Details of the proof are omitted here.
\begin{remark}
Initial efficiency is the essential supremum of $\beta^+$.
\end{remark}

\section{Discussion}
The evidence points to the principle that rate splitting is optimal for product resources asymptotically in most coding problems admitting single-letter information theoretic solutions. Indeed, the algebraic manipulations in the converse proofs usually rely only on the independence of $\{X_t\}$, rather than that they are identically distributed. Hence the main element in proving such a result about rate splitting (e.g. Lemma \ref{lem1} in the appendix) is usually related to the converse proof of the corresponding coding theorem. However, this is not the case for coding problems of combinatorial nature. For example, the zero-error capacity of independent channels operating in parallel is not the sum of the zero-error capacities of the individual channels \cite{alon1998shannon}.


\section*{Acknowledgment}
This work was supported by NSF under Grants CCF-1116013, CCF-1319299, CCF-1319304, and the Air Force Office of Scientific Research under Grant FA9550-12-1-0196.

\appendix
\begin{lemma}\label{lem1}
Suppose that $\{(X_i,Y_i,Z_i)\}_{i=1}^n$ possess the product structure of (\ref{assump}), and $(U,V)$ are r.v.'s such that $(U,V)-X_1^n-(Y_1^n,Z_1^n)$. Then there exist $U_1^n$ and $V_1^n$ such that $(U_i,V_i)-X_i-(Y_i,Z_i)$ for $i=1,\dots,n$ and
\begin{align}
I(U,V;X_1^n)-I(U,V;Y_1^n)\ge&\sum_{i=1}^n [I(U_i,V_i;X_i)-I(U_i,V_i;Y_i)],\label{e43}\\
I(V;Y_1^n|U)-I(V;Z_1^n|U)=&\sum_{i=1}^n [I(V_i;Y_i|U_i)-I(V_i;Z_i|U_i)].\label{e1}
\end{align}
\end{lemma}
\begin{proof}
%
Choose a random vector $\bar{Z}^n$ such that
\begin{align}
&P_{UVX^nY^n\bar{Z}^n}(u,v,x^n,y^n,z^n)\nonumber\\
=&P_{UVX^n}(u,v,x^n)P_{Y^n|X^n}(y^n|x^n)P_{Z^n|X^n}(z^n|x^n).
\end{align}
Define $\bar{U}_i=(Y^{i-1},\bar{Z}_{i+1}^n,U)$ and $V_i=V$. Then $(\bar{U}_i,V_i)-X_i-(Y_i,\bar{Z}_i)$ forms a Markov chain for each $i$. Moreover
\begin{align}\label{eq184}
I(V;Y_1^n|U)-I(V;\bar{Z}_1^n|U)=&\sum_{i=1}^n [I(V_i;Y_i|\bar{U}_i)-I(V_i;\bar{Z}_i|\bar{U}_i)]
\end{align}
holds, c.f. \cite[Lemma 4.1]{ahlswede1993common}. Next, observe that
\begin{align}
&I(U,V;X_1^n)-I(U,V;Y_1^n)\nonumber\\
=&\sum_{i=1}^n[I(U,V;X_i|X_{i+1}^n,Y^{i-1})-I(U,V;Y_i|X_{i+1}^nY^{i-1})]\label{st1}\\
=&\sum_{i=1}^n[I(U,V,X_{i+1}^n,Y^{i-1};X_i)-I(U,V,X_{i+1}^nY^{i-1};Y_i)]\label{st2}\\
=&\sum_{i=1}^n[I(U,V,X_{i+1}^n,Y^{i-1},\bar{Z}^n_{i+1};X_i)\nonumber\\
&\quad\quad\quad-I(U,V,X_{i+1}^nY^{i-1},\bar{Z}^n_{i+1};Y_i)]\label{st3}
\end{align}
\begin{align}
=&\sum_{i=1}^n[I(U,V,Y^{i-1},\bar{Z}^n_{i+1};X_i)-I(U,V,Y^{i-1},\bar{Z}^n_{i+1};Y_i)]\nonumber\\
&+\sum_{i=1}^n[I(X_{i+1}^n;X_i|U,V,Y^{i-1},\bar{Z}^n_{i+1})\nonumber\\
&\quad\quad\quad-I(X_{i+1}^n;Y_i|U,V,Y^{i-1},\bar{Z}^n_{i+1})]\\
\ge&\sum_{i=1}^n[I(U,V,Y^{i-1},\bar{Z}^n_{i+1};X_i)-
I(U,V,Y^{i-1},\bar{Z}^n_{i+1};Y_i)]\label{st4}\\
=&\sum_{i=1}^n[I(\bar{U}_i,V_i;X_i)-I(\bar{U}_i,V_i;Y_i)].\label{eq190}
\end{align}
Here, (\ref{st1}) is again by \cite[Lemma 4.1]{ahlswede1993common}, and (\ref{st2}) is from $(X_i,Y_i)\perp(X_{i+1}^n,Y^{i-1})$. Equality (\ref{st3}) is from $\bar{Z}_{i+1}^n-(U,V,X_{i+1}^n,Y^{i-1})-(X_i,Y_i)$.
Inequality (\ref{st4}) is because of $I(X_{i+1}^n;X_i|U,V,Y^{i-1},\bar{Z}^n_{i+1})=I(X_{i+1}^n;X_i,Y_i|U,V,Y^{i-1},\bar{Z}^n_{i+1})$.
Finally choose a random vector $U^n$ satisfying
\begin{align}
P_{U_i|V_iX_iY_iZ_i}(u_i|v_i,x_i,y_i,z_i)=P_{\bar{U}_i|V_iX_i}(u_i|v_ix_i).
\end{align}
Since $P_{U_iV_iX_iZ_i}=P_{\bar{U}_iV_iX_i\bar{Z}_i}$, we have
\begin{align}
I(V_i;Z_i|U_i)=&I(V_i;\bar{Z}_i|\bar{U}_i).\label{e185}
\end{align}
By the same token, we also have
\begin{align}
I(V_i;Y_i|U_i)=&I(V_i;Y_i|\bar{U}_i),\quad
I(U_i,V_i;Y_i)=&I(\bar{U}_i,V_i;Y_i),
\\
I(U_i,V_i;X_i)=&I(\bar{U}_i,V_i;X_i),\quad
I(V;Z_1^n|U)=&I(V;\bar{Z}_1^n|U).\label{e184}
\end{align}
The last five identities and \eqref{eq184}, \eqref{eq190} complete the proof.
\end{proof}
\bibliographystyle{ieeetr}
\bibliography{ref2014}

\begin{thebibliography}{10}

\bibitem{csiszar2004secrecy}
I.~Csisz{\'a}r and P.~Narayan, ``Secrecy capacities for multiple terminals,''
  {\em IEEE Transactions on Information Theory}, vol.~50, no.~12,
  pp.~3047--3061, 2004.

\bibitem{ahlswede1998common}
R.~Ahlswede and I.~Csisz{\'a}r, ``Common randomness in information theory and
  cryptography. ii. cr capacity,'' {\em IEEE Transactions on Information
  Theory}, vol.~44, no.~1, pp.~225--240, 1998.

\bibitem{watanabe}
S.~Watanabe and Y.~Oohama, ``Secret key agreement from vector gaussian sources
  by rate limited public communication,'' {\em Information Theory Proceedings
  (ISIT), 2010 IEEE International Symposium on}, pp.~2597--2601, 2010.

\bibitem{shannon1959coding}
C.~E. Shannon, ``Coding theorems for a discrete source with a fidelity
  criterion,'' {\em IRE Nat. Conv. Rec}, vol.~4, no.~142-163, 1959.

\bibitem{cover2012elements}
T.~M. Cover and J.~A. Thomas, {\em Elements of Information Theory}.
\newblock John Wiley \& Sons, 2012.

\bibitem{anantharam2013maximal}
V.~Anantharam, A.~Gohari, S.~Kamath, and C.~Nair, ``On maximal correlation,
  hypercontractivity, and the data processing inequality studied by erkip and
  cover,'' {\em arXiv preprint arXiv:1304.6133}, 2013.

\bibitem{erkip1998efficiency}
E.~Erkip and T.~M. Cover, ``The efficiency of investment information,'' {\em
  IEEE Transactions on Information Theory}, vol.~44, no.~3, pp.~1026--1040,
  1998.

\bibitem{witsenhausen1975sequences}
H.~S. Witsenhausen, ``On sequences of pairs of dependent random variables,''
  {\em SIAM Journal on Applied Mathematics}, vol.~28, no.~1, pp.~100--113,
  1975.

\bibitem{zhao}
L.~Zhao, ``Common randomness, efficiency, and actions,'' {\em PhD thesis,
  Stanford University}.

\bibitem{gray2006toeplitz}
R.~M. Gray, {\em Toeplitz and circulant matrices: A review}.
\newblock now publishers Inc, 2006.

\bibitem{alon1998shannon}
N.~Alon, ``The {Shannon} capacity of a union,'' {\em Combinatorica}, vol.~18,
  no.~3, pp.~301--310, 1998.

\bibitem{ahlswede1993common}
R.~Ahlswede and I.~Csisz\'{a}r, ``Common randomness in information theory and
  cryptography. i. secret sharing,'' {\em IEEE Transactions on Information
  Theory}, vol.~39, no.~4, pp.~1121--1132, 1993.

\end{thebibliography}
\end{document}